\documentclass[orivec]{llncs}

\usepackage[utf8]{inputenc} 
\usepackage{theorem}
\usepackage{amssymb}
\usepackage{amsmath}
\usepackage{subfigure}
\usepackage{tikz}
\usepackage{fix me}
\usepackage{amsfonts,pifont}
\usepackage{MnSymbol}
\usepackage{hyperref}
\usepackage{xcolor}
\usepackage{xspace}
\usepackage{multicol}
\usepackage[center,width=152mm,height=235mm]{crop}
\newtheorem{assumption}{\normalfont \textbf{Assumption}}

\newcommand{\ind}{\meddiamond}

\newcommand{\indt}[1]{\mathrel{\ind_{#1}}}
\newcommand{\indu}[1]{\mathrel{\wasylozenge_{#1}}}
\newcommand{\dep}{\diamondtimes}
\newcommand{\reach}[1]{\mathop{\mathit{reach}} (#1)}
\newcommand{\obs}[1]{\mathop{\mathit{obs}} (#1)}
\newcommand{\preset}[1]{{{}^\bullet{#1}}}
\newcommand{\postset}[1]{{#1}^\bullet}

\newcommand{\leqi}{\lessdot}
\newcommand{\logs}{\mathcal{L}}

\newcommand{\pnet}{\mathcal{N}}
\newcommand{\sys}{\mathcal{S}}
\newcommand{\model}{\mathcal{M}}
\newcommand{\net}{N}

\newcommand{\netdef}{\net \eqdef (P,T,F)}
\newcommand{\pnetdef}{\pnet \eqdef (P,T,F,\lambda,M_0)}
\newcommand{\netl}{\onl^\sim}

\newcommand{\on}{\beta}
\newcommand{\ondef}{\beta \eqdef (B,E,F)}
\newcommand{\onl}{\on_{\logs,\ind}}
\newcommand{\onls}{\on_{\logs,\ind,*}}
\newcommand{\netls}{\onls^\sim}
\newcommand{\les}{\mathcal{E}}
\newcommand{\lesdef}{\les \eqdef (E,\leq,\#)}
\newcommand{\coe}[1]{\mathop{\mathit{coe}}(#1)}
\newcommand{\co}{\textbf{ co }}
\newcommand{\es}[1]{ES(#1)}
\newcommand{\lpo}[1]{lpo_\ind(#1)}
\newcommand{\arrow}[1]{\mathrel{\raisebox{-1.1pt}{$\xrightarrow{#1}$}}}
\newcommand\nat{\mathbb{N}}
\newcommand\esig{e_{\sigma}}
\newcommand{\ssc}[2]{\phi_{#1,#2}^{sub}}
\newcommand{\disj}[2]{\phi_{#1,#2}^{disj}}

\newcommand{\marking}[1]{\mathop{\mathit{Mark}}(#1)}

\newcommand\podtool  {\textsc{Pod}\@\xspace}
\newcommand\eqdef    {\mathrel{:=}}
\newcommand\tup[1]   {\langle#1\rangle}
\newcommand\set[1]   {{\{ #1 \mathclose \}}}

\newcommand\cfl[1][] {\mathrel{\#_{#1}}}
\newcommand\dcfl     {\cfl[d]}

\newcommand\wrt      {w.r.t.\@\xspace}

\newcommand\eg       {e.g.\@\xspace}

\newcommand\ie       {i.e.\@\xspace}

\newcommand\cols     {cols.\@\xspace}

\newcommand\bench[2][]{%
\if\relax\detokenize{#1}\relax%
\textsc{#2}\else \textsc{#2}\hspace{0.3pt}{\footnotesize(}#1{\footnotesize)}\fi\@\xspace}

\usepackage{etoolbox}
\usepackage{thmtools}
\usepackage{thm-restate}
\usepackage{nicefrac}
\usepackage{booktabs}
\usepackage{wasysym}

\providetoggle{long}
\settoggle{long}{true}

\hypersetup{
  bookmarksdepth=2,
  bookmarksnumbered=true,
  bookmarksopen=true,
  bookmarksopenlevel=2,
  colorlinks=true,
  linktocpage=true,
  breaklinks=true,
  pageanchor=true,
  allcolors=[rgb]{0.6,0.0,0.0},
  pdftitle={Unfolding-Based Process Discovery},
  pdfauthor={Hernán Ponce-de-León, César Rodríguez, Josep Carmona, Keijo
  Heljanko, Stefan Haar}
}

\ifdefined\VersionWithComments
\newcommand{\cro}[1]{{\color{teal}\textbf{[César}: #1]}}

\newcommand{\todo}[1]{{\color{red!90!black}\textbf{[TODO}: #1]}}
\else
\newcommand{\cro}[1]{}
\newcommand{\ms}[1]{}
\newcommand{\svs}[1]{}
\newcommand{\todo}[1]{}
\fi

\usetikzlibrary{automata,petri,positioning,shapes,decorations,arrows,backgrounds}
\tikzstyle{tplace}=[circle,draw,inner sep=1.5mm]

\title{Unfolding-Based Process Discovery\texorpdfstring{\thanks{This 
is the unabridged version of a paper with the same title appeared at the
proceedings of ATVA~2015.}}{}}

\author{Hernán~Ponce-de-Le\'on\inst{1} \and C\'esar Rodr\'iguez\inst{2} \and Josep~Carmona\inst{3} \and Keijo Heljanko\inst{1} \and Stefan Haar\inst{4}}
\institute{
Helsinki Institute for Information Technology HIIT and Department of Computer Science and Engineering, School of Science, Aalto University, Finland \\
\email{\{hernan.poncedeleon,keijo.heljanko\}@aalto.fi} \\
\and
Universit\'e Paris 13, Sorbonne Paris Cit\'e, LIPN, CNRS, France \\
\email{cesar.rodriguez@lipn.fr}
\and
Universitat Polit\`ecnica de Catalunya, Barcelona, Spain \\
\email{jcarmona@cs.upc.edu} \\
\and
INRIA and LSV, \'Ecole Normale Sup\'erieure de Cachan and CNRS, France \\
\email{stefan.haar@inria.fr}
}

\begin{document}

\maketitle

\begin{abstract}
This paper presents a novel technique for process discovery.
In contrast to the current trend, which only
considers an event log for discovering a process model,
we assume two additional inputs: an independence
relation on the set of logged activities, and a collection of negative traces.
After deriving an intermediate net unfolding from them,
we perform a controlled folding giving rise to a Petri net which
contains both the input log and all independence-equivalent traces arising from it.
Remarkably, the
derived Petri net cannot execute any trace from the negative collection.
The entire chain of transformations is fully automated.
A tool has been developed and experimental results are provided that
witness the significance of the contribution of this paper.
\end{abstract}

\section{Introduction}

The derivation of process models from partial observations has received significant attention in the last
years, as it enables eliciting evidence-based formal representations of the real processes running in a
system~\cite{AalstBook}. This discipline, known as \emph{process discovery}, has similar premises as in
\emph{regression analysis}, i.e., only when moderate assumptions are made on the input data one can derive faithful
models that represent the underlying system.

Formally, a technique for process discovery receives as input an
\emph{event log}, containing the footprints of a process' executions,
and produces a model (\eg, a Petri net) describing the real process.
Many process discovery algorithms in the literature
make strong implicit assumptions.
A widely used 
one is \emph{log completeness}, requiring every possible trace of the underlying
system to be contained in the event log.
This is hard to satisfy by systems with cyclic or infinite
behavior, but also for systems that evolve continuously over time. Another implicit assumption is the lack of
\emph{noise} in the log, \ie,
traces denoting exceptional behavior that should not be contained in the derived
process model.
Finally, every discovery technique has a \emph{representational bias}.
For instance, the $\alpha$-algorithm~\cite{Aalst11} can only discover Petri nets of a specific
class (\emph{structured workflow nets}).

Few attempts have been made to remove the aforementioned assumptions.
One promising direction is to relieve 
the discovery problem by assuming that
more knowledge about the underlying system is available as input.
On this line, the works
in~\cite{Ferreira2006,Lamma2008,Goedertier2009} are among the few that use
domain knowledge in terms of
\emph{negative information}, expressed by traces which do not represent process behavior.
In this paper we follow 
this direction, but additionally incorporate a crucial information to be used for the task of process discovery: when a
pair of activities are {\em independent} of each other. One example could be the
different tests that a patient should
undergo in order to have a diagnosis: blood test, allergy test, and radiology
test, which are independent each other.
We believe that obtaining this coarse-grain independence information from a domain expert is an easy and natural step;
however, if they are not available, one can estimate them from analysing the log with some of the techniques in the
literature, e.g., the relations computed by the $\alpha$-algorithm~\cite{AalstBook}.

\begin{figure}[t]
\begin{centering}
\includegraphics[scale=0.6]{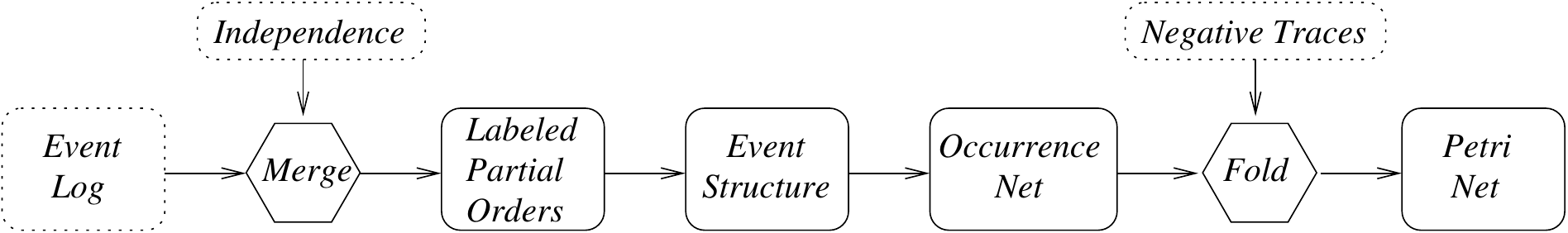}
\caption{\label{fig:flow} Unfolding-based process discovery.}
\end{centering}
\end{figure}

The approach of this paper is summarized in \autoref{fig:flow}.
Starting from an event log and an independence
relation on its set of activities, we conceptually construct
a collection of \emph{labeled partial orders} whose
linearizations include both the sequences in the log as well those in the same
Mazurkiewicz trace~\cite{Mazurkiewicz86},
\ie, those obtained via successive permutations of independent activities.
We then merge (the common prefixes of) this collection into an
\emph{event structure} which we next transform into an occurrence net
representing the same behavior. Finally, we perform a controlled generalization
by selectively folding the occurrence net into a Petri net.
This step yields a net that
(a) can execute all traces contained in the event log, and
(b) generalizes the behavior of the log in a controlled manner, introducing no
execution given in the collection of negative traces.
The folding process is driven by a \emph{folding equivalence relation}, which we
synthesize using SMT.
Different folding equivalences guarantee different properties about the final net.
The paper proposes three different classes of equivalences and
studies their properties.
In particular we define a class of \emph{independence-preserving} folding
equivalences, guaranteeing that the natural independence relation in the final
net will equal the one given by the expert.

In summary, the main contributions of the paper are:
\begin{itemize}
\item
  A general and efficient translation from prime event structures to occurrence
  nets 
  (\autoref{sec:discovery}).
\item
  Three classes of folding equivalences of interest not only in process
  discovery but also in formal verification of concurrent systems
  (\autoref{sec:generalization}).
\item
  A method to synthesize folding equivalences using SMT
  (\autoref{sec:computing}).
\item
  An implementation of our approach and experimental results witnessing its
  capacity to even rediscover the original model
  (\autoref{sec:experiments}).
\end{itemize}

Remarkably, the discovery technique of this paper solves for the first time one of the foreseen operations 
in~\cite{DumasG15},
which advocates for the unified use of event structures to support process mining operations.


\section{Preliminaries}
\label{sec:prelim}


\paragraph{Events:}
given an alphabet of actions $A$, several occurrences of a given action can
happen on a run or execution.
In this paper we consider a set~$E$ of events
representing the occurrence of
actions in executions. Each event $e \in E$ has the form $e\eqdef\tup{a,H}$,
where $a \in A$ and~$H \subseteq E$ is a subset of events causing~$e$ (its history).
The label of an event is given by a function $\lambda \colon E \to A$ defined as
$\lambda(\langle a,H \rangle) \eqdef a$.

\paragraph{Labeled partial orders (lpos):}
we represent a labelled partial order by a pair $(E, \le)$, where ${\le} \subseteq E \times E$ is a reflexive, antisymmetric and transitive
relation on the set~$E$ of events.
Two distinct
events $e,e' \in E$ can be either ordered ($e \le e'$ or $e' \le e$) or concurrent ($e \not\le e'$ and $e' \not\le e$).
Observe that all events are implicitly labelled by $\lambda$.

\paragraph{Petri nets:}

a net consists of two disjoint sets $P$ and $T$ representing respectively places
and transitions together with a set $F$ of flow arcs. The notion of state of the
system in a net is captured by its markings.
A marking is a multiset $M$ of
places, \ie, a map $M \colon P \to \nat$.
We focus on the so-called safe nets, where markings are sets, \ie, $M(p) \in \{0,1\}$
for all $p \in P$. A Petri net (PN) is a net together with an initial marking and a total function that labels its
transitions over an alphabet $A$ of observable actions. Formally a PN is a tuple $\pnetdef$ where \emph{(i)} $P \not =
\emptyset$ is a set of places; \emph{(ii)} $T \not = \emptyset$ is a set of transitions such that $P \cap T =
\emptyset$; \emph{(iii)} $F \subseteq (P \times T) \cup (T \times P)$ is a set
of flow arcs; \emph{(iv)} $\lambda \colon T \to A$ is a labeling mapping;
and \emph{(v)} $M_0 \subseteq P$ is an initial marking. Elements of $P \cup T$ are called the
nodes of $\pnet$. For a transition  $t \in T$, we call $\preset t \eqdef \{ p \mid (p,t) \in 
F \}$ the preset of $t$, and $\postset t \eqdef \{p \mid (t,p) \in F \}$ the postset of $t$. In figures, we represent as
usual places by empty circles, transitions by squares, $F$ by arrows, and the marking of a place $p$ by  black tokens in
$p$. A transition $t$ is enabled in marking $M$, written $M \arrow t$, iff $\preset t \subseteq M$. This enabled
transition can fire, resulting in a new marking $M' \eqdef (M \backslash \preset t) \cup \postset t$. This firing relation is
denoted by $M \arrow t M'$. A marking $M$ is reachable from $M_0$ if there exists a firing sequence, i.e.  transitions,
$t_1, \dots, t_n$ such that $M_0 \arrow{t_1} \dots \arrow{t_n} M$. The set of reachable markings from $M_0$ is denoted by
$\reach \pnet$. The set of co-enabled transitions is
${\coe \pnet} \eqdef
\{ (t,t') \mid \exists M \in \reach \pnet \colon \preset t \subseteq M \land \preset{t'} \subseteq M \}$.
The set of observations of a net
is the image over $\lambda$ of its
fireable sequences, \ie,
$\sigma \in \obs \pnet$ iff $M_0 \arrow{t_1} \dots \arrow{t_n} M$ and
$\lambda(t_1) \dots \lambda(t_n) = \sigma$.

\paragraph{Occurrence nets:}

occurrence nets can be seen as infinite Petri nets with a special acyclic
structure that highlights conflict between transitions that compete for
resources. Places and transitions of an occurrence net are usually called
conditions and events. Formally, let $\netdef$ be a net,  $<$  the transitive
closure of $F$, and $\leq$ the reflexive closure of $<$. We say that transitions
$t_1$ and $t_2$ are in structural conflict,
written $t_1 \cfl[s] t_2$, if and only
if $t_1 \not = t_2$ and $\preset{t_1} \cap \preset{t_2} \not = \emptyset$.
Conflict is inherited along $<$, that is, the conflict relation $\cfl$ is given by
$a \cfl b \Leftrightarrow \exists t_a,t_b \in T \colon t_a \cfl[s] t_b \land t_a
\leq a \land t_b \leq b$. Finally, the concurrency relation $\bf co$ holds
between nodes $a,b \in P \cup T$ that are neither ordered nor in conflict, i.e.
$a\co b \Leftrightarrow \neg (a \leq b) \land \neg (a \cfl b) \land \neg (b \leq a)$.

A net $\ondef$ is an occurrence net iff \emph{(i)} $\leq$ is a partial order; \emph{(ii)} for all $b \in B$, $\lvert
\preset b \rvert \in \{0,1\}$; \emph{(iii)} for all $x\in B \cup E$, the set $[x] := \{y \in E \mid y \leq x\}$ is finite;
\emph{(iv)} there is no self-conflict, i.e. there is no $x \in B \cup E$ such
that $x\cfl x$. The initial marking $M_0$ of an occurrence net is the set of
conditions with an empty preset, i.e. $\forall b \in B\colon b \in M_0
\Leftrightarrow \preset b = \emptyset$. Every $\leq$-closed and conflict-free
set of events $C$ is called a configuration and generates a reachable marking defined as
$\marking C \eqdef (M_0 \cup \postset C) \setminus \preset C$.
We also assume a labeling
function $\lambda \colon E \to A$ from events in $\beta$ to alphabet $A$. 
Conditions are of the form $\langle e , X \rangle$ where $e \in E$ is the event generating the condition and $X
\subseteq E$ are the events consuming it. Occurrence nets are the mathematical form of the partial order unfolding
semantics of a Petri net~\cite{EsparzaRV02}; we use indifferently the terms occurrence net and unfolding.

\iftoggle{long}{
\begin{lemma}
\label{lem:cocond}
  Let $\on$ be an occurrence net such that $(e,e') \in {\coe \on}$, then for every $b \in \preset e$ and $b' \in \preset{e'}$ we have $b=b'$ or $b \co b'$.
\end{lemma}

\begin{proof}
Routine.
\end{proof}
}

Conditions in an occurrence net can be removed by keeping the causal
dependencies and introducing a conflict relation; the obtained object is an event structure~\cite{NielsenPW81}.

\paragraph{Event structures:}

an event structure is a tuple $\lesdef$ where $E$ is a set of events; $\leq\ \subseteq E \times E$ is a partial
order (called causality) satisfying the property of finite causes, i.e. $\forall e \in E : \lvert [e] \rvert < \infty$ where $[e] := \{ e' \in E \mid e' \leq e
\}$; ${\cfl} \subseteq E \times E$ is an irreflexive symmetric relation (called conflict) satisfying the property
of conflict heredity, i.e. $\forall e,e',e'' \in E : e \cfl e' \land e' \leq e''
\Rightarrow e \cfl e''$. Note that in most cases one only needs to consider
reduced versions of relations $\leq$ and $\cfl$, which we will denote $\leqi$ and
$\dcfl$, respectively. Formally, $\leqi$ (which we call direct causality) is the
transitive reduction of $\leq$, and $\dcfl$ (direct conflict) is the smallest
relation inducing $\cfl$ through the property of conflict heredity. A
configuration is a computation state represented by a set of events that have occurred; 
if an event is present in a configuration, then so must all the events on which it causally depends.
Moreover, a configuration does not contain conflicting events. Formally, a
configuration of $(E,{\leq},{\cfl})$ is a set
$C \subseteq E$ such that $e \in C \Rightarrow (\forall e' \leq e : e' \in C)$,
and $(e \in C \land e \cfl e') \Rightarrow e'\not\in C$.
The set of 
configurations of~$\les$ is denoted by~$\Omega(\les)$.

\paragraph{Mazurkiewicz traces:}

let $A$ be a finite alphabet of letters and $\ind \subseteq A \times A$ a symmetric and irreflexive relation called
independence. The relation $\ind$ induces an equivalence relation $\equiv_\ind$ over
$A^*$. Two words $\sigma$ and $\sigma'$ are equivalent ($\sigma \equiv_\ind \sigma'$) if there exists a sequence
$\sigma_1 \dots \sigma_k$ of words such that $\sigma=\sigma_1, \sigma'=\sigma_k$ and for all $1\leq i \leq k$ there
exists words $\sigma_i', \sigma_i''$ and letters $a_i,b_i$ satisfying $$\sigma_i=\sigma_i' a_i b_i \sigma_i'',
\hspace{5mm} \sigma_{i+1}=\sigma_i' b_i a_i \sigma_i'', \hspace{4mm} \text{and } (a_i,b_i) \in \ind$$ Thus, two words
are equivalent by $\equiv_\ind$ if one can be obtained from the other by successive commutation of neighboring
independent letters. For a word $\sigma \in A^*$ the equivalence class of $\sigma$ under $\equiv_\ind$ is called a
Mazurkiewicz trace~\cite{Mazurkiewicz86}.
\bigskip 

We now describe the problem tackled in this paper, one of the main challenges in the {\em process mining}
field~\cite{AalstBook}.

\paragraph{Process Discovery:}

a log ${\logs}$ is a finite set of traces over an alphabet $A$ representing the footprints of the real process
executions of a system $\sys$ that is only (partially) visible through these runs. Process discovery techniques aim
at extracting from a log ${\logs}$ a process model ${\model}$ (e.g., a Petri net) with the goal to elicit the process
underlying in ${\sys}$. By relating the behaviors of ${\logs}$, $\obs {\model}$ and ${\sys}$,
particular concepts can be defined~\cite{BuijsDA14}. A log is \emph{incomplete} if ${\sys} \backslash {\logs} \ne
\emptyset$. A model ${\model}$ \emph{fits} log ${\logs}$ if ${\logs} \subseteq \obs {\model}$. A model is
\emph{precise} in describing a log ${\logs}$ if $\obs {\model} \backslash {\logs}$ is small. A model ${\model}$
represents a \emph{generalization} of log ${\logs}$ with respect to system ${\sys}$ if some behavior in ${\sys}
\backslash {\logs}$ exists in $\obs {\model}$. Finally, a model ${\model}$ is \emph{simple} when it has the minimal
complexity in representing $\obs {\model}$, i.e., the well-known \emph{Occam's razor principle}.
It is widely acknowledged that the size of a process model is the
most important simplicity indicator. Let ${\cal U}^{\cal N}$ be the universe
of nets, we define a function $\hat c: {\cal U}^{\cal N} \to \mathbb{N}$ to
measure the simplicity of a net by counting the number of some of its elements, \eg, its
transitions and/or places.

\section{Independence-Preserving Discovery}
\label{sec:discovery}

Let $\sys$ be a system whose set of actions is~$A$.
Given two actions $a,b \in A$ and one state~$s$ of~$\sys$,
we say that~$a$ and~$b$ \emph{commute} at~$s$ when
\begin{itemize}
\item
  if $a$ can fire at $s$ and its execution reaches state~$s'$,
  then $b$ is possible at~$s$ iff it is possible at~$s'$; and
\item
  if both $a$ and $b$ can fire at~$s$, then
  firing $ab$ and $ba$ reaches the same state.
\end{itemize}
Commutativity of actions at states identifies an equivalence relation in the
set of executions of the system $\sys$; 
it is a \emph{ternary} relation, relating two transitions with one state.

Since asking the expert to provide the commutativity relation of~$\sys$
would be difficult, we restrict ourselves to
unconditional independence, \ie,
a conservative overapproximation of the commutativity relation that is a sole
property of transitions, as opposed to transitions and states.
An \emph{unconditional independence} relation of~$\sys$
is any \emph{binary}, symmetric, and irreflexive relation~$\ind \subseteq A \times A$
satisfying that if~$a \ind b$ then~$a$ and~$b$ commute at
\emph{every reachable state} of~$\sys$.
If $a, b$ are not independent according to $\ind$,
then they are dependent, denoted by $a \dep b$.

In this section,
given a log~$\logs \subseteq A^*$, representing some behaviors of~$\sys$,
and an arbitrary unconditional independence~$\ind$ of~$\sys$,
provided by the expert, we construct an occurrence net
whose executions contain $\logs$ together with
all sequences in $A^*$ which are $\equiv_{\ind}$-equivalent to some sequence
in~$\logs$.

If commuting actions are not declared independent by the expert
(\ie, $\ind$~is smaller than it could be),
then $\model$ will be more sequential than $\sys$;
if some actions that did not commute are marked as independent,
then $\model$ will not be a truthful representation of $\sys$.
The use of expert knowledge in terms of an independence relation is a novel
feature not considered before in the context of process discovery. We believe
this is a powerful way to fight
with the problem of log incompleteness in a practical way since it is only needed to observe in the log one
trace representative of a class in $\equiv_\ind$ to include the whole set of traces of the class in the process model's
executions.

Our final goal is to generate a Petri net that represents the behavior of the
underlying system. We start by translating~$\logs$ into a collection of partial
orders whose shape depends on the specific definition of~$\ind$.

\begin{definition}
\label{def:log2lpo}
  Given a sequence $\sigma \in A^*$ and an independence relation $\ind \subseteq
  A \times A$, we associate to~$\sigma$ a labeled partial order $\lpo \sigma$
  inductively defined by:
  \begin{enumerate}
    \item
      If $\sigma = \varepsilon$, then let $\bot \eqdef \tup{\tau, \emptyset}$
      and set $\lpo \sigma \eqdef (\set \bot, \emptyset)$.
    \item
      If $\sigma = \sigma' a$,
      then let $\lpo{\sigma'} \eqdef (E',\leq')$ and
      let $e \eqdef \tup{a, H}$ be the single event such that~$H$ is the unique
      $\subseteq$-minimal, causally-closed set of events in~$E'$
      satisfying that for any event $e' \in E'$,
      if $\lambda(e') \dep a$, then $e' \in H$.
      Then set $\lpo \sigma \eqdef (E,\leq)$ with $E \eqdef E' \cup \{ e \}$ and
      ${\leq} \eqdef {\leq'} \cup (H \times \{ e \})$.
  \end{enumerate}
\end{definition}

Since a system rarely generates a single observation, we need a compact way to model all the possible observations of the system. We represent all the partially ordered executions of a system with an event structure.

\begin{definition}
\label{def:lpo2es}
  Given a set of partial orders $ S \eqdef \{ (E_i,\leq_i) \mid 1 \leq i \leq n \}$,
  we define $\es S \eqdef (E,\leq, \cfl)$
where:
  \begin{enumerate}
    \item $E \eqdef \bigcup\limits_{1 \leq i \leq n} E_i$,
    \item ${\leq} \eqdef (\bigcup\limits_{1 \leq i \leq n} \leq_i)^*$, and
    \item
      for $e \eqdef \langle a,H \rangle$ and $e' \eqdef \langle b,H' \rangle$, we have that
      $e \dcfl e'$ (read: $e$ and $e'$ are in direct conflict) iff
      $e' \not \in H, e \not \in H'$ and $a \dep b$.
      The conflict relation  $\cfl$ is the
      smallest relation that includes $\dcfl$ and is inherited \wrt $\leq$,
      \ie, for $e \cfl e'$ and $e \leq f$, $e' \leq f'$, one has $f \cfl f'$.    
\end{enumerate}
\end{definition}

\iftoggle{long}{
\begin{lemma}
\label{le:lpo2es}
$\es S\eqdef (E,\leq, \cfl)$ from \autoref{def:lpo2es} is an event structure.
\end{lemma}

\begin{proof}
Clearly, $\leq$ is reflexive, antisymmetric and transitive by the Kleene closure: violations to these properties on
$\leq$ will contradict the corresponding properties in some lpo in $S$ since every event is characterized by the
set of their causal events.
Now, the definition of $e \dcfl e'$ is clearly symmetric in the roles of $e$ and
$e'$ since $\dep$ is a symmetric relation; symmetry is also inherited under $\leq$. 
\end{proof}

\begin{lemma}
  $\es S\eqdef (E,\leq, \cfl)$ from \autoref{def:lpo2es} is unique.
\end{lemma}

\begin{proof}
The set $E$ and the relation $\leq$ are clearly unique since they are defined from the union and Kleene closure
operators, respectively, which derive unique results. Now, $\dcfl$ is unique since its definition in
\autoref{def:lpo2es} is based on removing all the causality from $\dep$, and hence only one possible relation is
obtained for $\dcfl$. By taking the smallest relation including $\dcfl$ that is inherited with respect $\leq$, again only
one possible relation is obtained for $\cfl$.
\end{proof}
}

Given a 
set of finite partial orders $S$, we now show that $S$ 
is included in the 
configurations of the event structure obtained by \autoref{def:lpo2es}. 
This means that our event structure is a fitting
representation of~$\logs$.

\begin{proposition}
\label{lemma:fit}
If $S$
is finite, then $S \subseteq \Omega(\es S)$.
\end{proposition}

\iftoggle{long}{
\begin{proof}
By \autoref{def:lpo2es}.1, all the events of the partial orders are part of the event structure. Let $(E_i,\leq_i) \in S$, clearly $E_i$ is casually closed in $\es S$ (\autoref{def:lpo2es}.2). By \autoref{def:lpo2es}.3, any event in conflict with some event $e \in E_i$ is not in its past; we can conclude that $E_i$ is conflict free and therefore a configuration.
%
\qed
\end{proof}
}

Since we want to produce a Petri net, we now need to
``\emph{attach conditions}'' to the result of \autoref{def:lpo2es}.
Event structures and occurrence nets are conceptually very similar
objects so this might seem very easy for the acquainted reader.
However, this definition is crucial for the success of the subsequent folding
step (\autoref{sec:generalization}), as we will be constrained to merge
conditions in the preset and postset of an event when we merge the event.
As a result, the conditions that we produce now should constraint as little as
possible the future folding step.

\begin{definition}
\label{def:es2on}
  Given an event structure $\lesdef$ we construct the occurrence net $\on \eqdef
  (B,E \backslash \{ \bot \},F)$ in two steps
\begin{enumerate}
  \item Let $G \eqdef (V,A)$ be a graph where $V \eqdef E$ and $(e_1,e_2) \in A$ iff $e_1
  \dcfl e_2$. For each clique (maximal complete subgraph) $K \eqdef \{e_1, \dots, e_n \}$ of $G$, let $C_K \eqdef [e_1] \cap \dots \cap [e_n]$ and $e_K \in \max (C_K)$. We add a condition $b$ to $B$ and set $b \in \postset{e_K}$ and $b \in \preset{e_i}$ for $i = 1 \dots n$.
  \item For each $e \in E$, let $G_e \eqdef (V_e,A_e)$ be a graph where $V_e \eqdef \{ e' \in E \mid e \leqi e' \}$ and $(e_1,e_2) \in A_e$ iff $\lambda(e_1) \dep \lambda(e_2)$. For each clique $K_e := \{e_1, \dots, e_n \}$ of $G_e$, we add a condition $b$ to $B$ and set $b \in \postset{e}$ and $b \in \preset{e_i}$ for $i = 1 \dots n$.
\end{enumerate}
\end{definition}

\autoref{def:es2on}.1. adds a condition for every set of pairwise direct conflicting events; the condition is generated by some event $e_K$ which is in the past of every conflicting event and consumed by all of them; by the latter the conflict of the event structure is preserved in the occurrence net. For each event and its immediate successors, \autoref{def:es2on}.2. adds conditions between them to preserve causality. To minimize the number of conditions, for the successor events having dependent labels only one condition is generated. This step does not introduce new conflicts in the occurrence net since the events have dependent labels and none is in the past of the other, then by \autoref{def:lpo2es} they are also in conflict in the event structure.

We note that Winskel already explained, in categorical terms,
how to relate an event structure with an occurrence net~\cite{Winskel84a}.
However, his definition is of interest only in that context,
while ours focus on a practical and efficient translation.


Given a log $\logs$ and an independence relation $\ind$, the net obtained
applying Definitions \ref{def:log2lpo}, \ref{def:lpo2es} and \ref{def:es2on}, in
this order, is denoted by $\onl$.
Since every trace in $\logs$ is a linearization of some of
the partial orders in the set $S$ obtained by \autoref{def:log2lpo} and 
these partial orders are included by \autoref{lemma:fit} in the configurations of $\es S$ (which are
the same as the configurations in $\onl$), the
obtained net is fitting.

\begin{proposition}
\label{prop:fit}
  Let $\logs$ be a log and $\ind$ an independence relation, for every $\sigma
  \in \logs$ we have $\sigma \in \obs \onl$.
\end{proposition}

\begin{proof}
  Since every trace is a linearization of some partial order obtained by \autoref{def:lpo2es}, by \autoref{lemma:fit} every trace is a linearization of the maximal configurations of the event structure; since causality and conflict are preserved by \autoref{def:es2on}, their configurations coincide, the trace correspond to a sequential execution of the occurrence net and the result holds.
\end{proof}

It is worth noticing that the obtained net generalizes the behavior of the
model, but in a controlled manner imposed by the independence relation.
For instance, if
$\logs \eqdef \set{ab}$ and $a \ind b$, then $ba \in \obs \onl$,
even if this behavior was not present in the log.
If the expert rightly declared~$a$ and~$b$ independent (\ie, if they commute at
all states of~$\sys$), then necessarily~$ba$ is a possible observation
of~$\sys$, even if it is not in~$\logs$.
The extra information provided by the expert allows us to 
generalize the discovered model in a provably sound manner,
thus coping with the log incompleteness problem.

\iftoggle{long}{
\begin{proposition}
\label{ass:cond_dep}
  Let $\onl \eqdef (B,E,F)$ be the unfolding obtained from the log $\logs$ with
  $\ind$ as the independence relation.
  For all pairs of events $e,e' \in E$ such that $(e,e') \in {\coe \onl}$ we have
  $\preset e \cap \preset{e'} \ne \emptyset
  \Leftrightarrow
  \lambda(e) \dep \lambda(e')$.
\end{proposition}

\begin{proof}
\item{$\Rightarrow$)} Let $b \in \preset e \cap \preset{e'}$; if $b$ was added
by \autoref{def:es2on}.2., the result trivially holds since the condition was
added in the preset of the events in the clique which relates only events with
dependent labels; if $b$ was added by \autoref{def:es2on}.1., then it was added
in the preset of the events in the clique which relates only direct conflicting
events and then $e \dcfl e'$ in the event structure. This means that none can be in the past of the other, because otherwise some event would be in self-conflict which is ruled out in event structures; now, by \autoref{def:lpo2es} we have $\lambda(e) \dep \lambda(e')$.
\item{$\Leftarrow$)} Let $\lambda(e) \dep \lambda(e')$, events $e$ and $e'$
could not be generated from the same log since if not they would be causally
related (see \autoref{def:log2lpo}) contradicting $(e,e') \in {\coe \onl}$; since
they were generated from different logs and have dependent labels, we have from
\autoref{def:lpo2es} that $e \dcfl e'$; since they are in conflict, \autoref{def:es2on}.1 adds a conditions in their presets and finally $\preset e \cap \preset{e'} \ne \emptyset$.
\qed
\end{proof}
}

The independence relation between labels gives rise to an arbitrary relation between transitions of a net
(not necessarily an independence relation):

\begin{definition}
\label{def:indt}
  Let $\ind \subseteq A \times A$ be an independence relation, $\pnetdef$ a net,
  and $\lambda \colon T \to A$.
  We define relation ${\indt \net} \subseteq T \times T$ between transitions
  of~$\net$ as
  $$
  t \indt \net t' \Leftrightarrow \lambda(t) \ind \lambda(t').
  $$
\end{definition}

In the next section we will define an approach to fold $\onl$ into a Petri
net whose natural independence relation equals $\ind$.
To formalize our approach we first need to define such natural independence.

\begin{definition}
\label{def:indu}
  Let $\netdef$ be a net.
  We define the \emph{natural independence} relation
  ${\indu \net} \subseteq T \times T$ on~$\net$ as
  $$t
  \indu \net t'
  \Leftrightarrow
  \preset t \cap \preset{t'} = \emptyset \land
  \postset t \cap \preset{t'} = \emptyset \land
  \preset t \cap \postset{t'} = \emptyset.
  $$
\end{definition}

In fact, one can prove that when~$\net$ is safe, then $\indu \net$ is the
notion of independence underlying the unfolding semantics of~$\net$.
In other words, the equivalence classes of $\equiv_{\indu \net}$ are in
bijective correspondence with the configurations in the unfolding of~$\net$.
The following result shows that the natural independence on the discovered
occurrence net corresponds to the relation provided by the expert,
when both we restrict to the set of co-enabled transitions.

\begin{theorem}
\label{prop:ind_on}
  Let $\onl$ be the occurrence net from the log $\logs$ with $\ind$ as the
  independence relation, then
  $$
  {\indt \onl} \cap {\coe \onl} = {\indu \onl} \cap {\coe \onl}
  $$
\end{theorem}

\iftoggle{long}{
\begin{proof}
  \item[$\subseteq)$] Let $(e,e') \in {\indt \onl} \cap {\coe \onl}$, then from \autoref{def:indt} follows that $\lambda(e) \ind \lambda(e')$ and by \autoref{ass:cond_dep} we have $\preset{e} \cap \preset{e'} = \emptyset$. Suppose $\postset{e} \cap \preset{e'} \not = \emptyset$ then $\exists b_1 \in \preset{e}$ such that $\forall b_2 \in \preset{e'}$ it holds that $b_1< b_2$ and by \autoref{lem:cocond} $(e,e') \not \in {\coe \onl}$ which leads to a contradiction. Using the same reasoning it can be proven that $\preset{e} \cap \postset{e'} = \emptyset$. By \autoref{def:indu} we can conclude that $(e,e') \in {\indu \onl} \cap {\coe \onl}$.
  \item[$\supseteq)$] Let $(e,e') \in {\indu \onl} \cap {\coe \onl}$, by \autoref{def:indu} we get $\preset{e} \cap \preset{e'} = \emptyset$ and since they are co-enabled, by \autoref{ass:cond_dep} follows $\lambda(e) \ind \lambda(e')$; finally by \autoref{def:indt} we have $(e,e') \in {\indt \onl}$ and since the events were co-enabled by assumption $(e,e') \in {\indt \onl} \cap {\coe \onl}$.
  \qed
\end{proof}
}

\section{Introducing Generalization}
\label{sec:generalization}

The construction described in the previous section guarantees that the unfolding obtained
is fitting (see \autoref{lemma:fit}). However, the difference between~$\sys$ and~$\logs$ may
be significant (e.g.,~$\sys$ can contain cyclic behavior that can be instantiated an arbitrary number of times whereas
only finite traces exist in~$\logs$) and the unfolding may be poor in generalization. 
The goal of this section is to generalize $\onl$ in a way that 
the right patterns from~$\sys$, partially observed in~$\logs$ (\eg, loops),
are incorporated in the generalized model.
To generalize, we fold the discovered occurrence net.
This folding is driven by an
%
%
%
equivalence relation~$\sim$ on~$E \cup B$ that dictates which events merge into
the same transition, and analogously for conditions; events cannot be merged with conditions.
We write $[x]_\sim \eqdef \{ x' \mid x \sim x' \}$ for the equivalence class of
node $x$. For a set $X$, $[X]_\sim \eqdef \{ [x]_\sim \mid x \in X \}$ is a set
of equivalence classes.

\begin{definition}[Folded net \cite{FahlandA13}]
\label{def:foldednet}
  Let $\on \eqdef (B,E,F)$ be an occurrence net and $\sim$ a equivalence relation on
  the nodes of $\on$.
  The folded Petri net (\wrt~$\sim$) is defined as
  $\on^\sim \eqdef (P_\sim,T_\sim,F_\sim,{M_0}_\sim)$ where
  \begin{align*}
  P_\sim & \eqdef \{ [b]_\sim \mid b \in B \}, &
  F_\sim & \eqdef \{ ([x]_\sim,[y]_\sim) \mid (x,y) \in F \}, \\
  T_\sim & \eqdef \{ [e]_\sim \mid e \in E \}, &
  {M_0}_\sim([b]_\sim) & \eqdef \lvert \{ b' \in [b]_\sim \mid \preset{b'} = \emptyset \} \rvert.
  \end{align*}
\end{definition}

Notice that the initial marking of the folded net is not necessarily safe. Safeness of the net depends on the chosen equivalence relation (see \autoref{prop:safe}). 

\subsection{Language-Preserving Generalization}
\label{sec:lpgeneralization}

Different folding equivalences guarantee different properties on the folded net.
From now on we focus our attention on three interesting classes of folding equivalences.
The first preserves sequential executions of $\onl$.

\begin{definition}[Sequence-preserving folding equivalence]
\label{def:fold1}
  Let $\on$ be an occurrence net; an equivalence relation $\sim$ is called a sequence preserving (SP) folding equivalence iff $e_1 \sim e_2$ implies $\lambda(e_1) = \lambda(e_2)$ and $[\preset{e_1}]_\sim = [\preset{e_2}]_\sim$ for all events $e_1,e_2 \in E$.
\end{definition}

From the definition above it follows that $e_1 \sim e_2$ implies $\forall b \in
\preset{e_1}: \exists b' \in \preset{e_2}$ with $b \sim b'$.  Since for every
folded net obtained from a SP folding equivalence only equally labeled events
are merged; we define then $\lambda([e]_\sim) \eqdef \lambda(e)$.

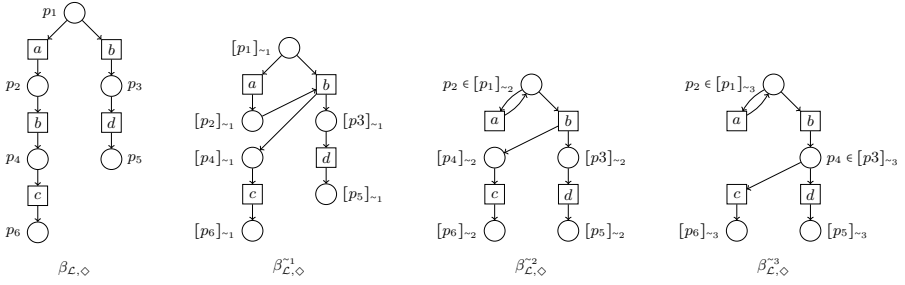
\begin{figure}[t]
\centering
  \subfigure{\scalebox{.65}{\begin{tikzpicture}

\node[tplace, label=left:$p_1$] (p1) at (0,-.75) {};
\node[transition] (e1) at (-.75,-1.5) {$a$};
\node[transition] (e2) at (.75,-1.5) {$b$};
\node[tplace, label=left:$p_2$] (p2) at (-.75,-2.25) {};
\node[tplace, label=right:$p_3$] (p3) at (.75,-2.25) {};
\node[transition] (e3) at (-.75,-3) {$b$};
\node[transition] (e4) at (.75,-3) {$d$};
\node[tplace, label=left:$p_4$] (p4) at (-.75,-3.75) {};
\node[tplace, label=right:$p_5$] (p5) at (.75,-3.75) {};
\node[transition] (e5) at (-.75,-4.5) {$c$};
\node[tplace, label=left:$p_6$] (p6) at (-.75,-5.25) {};

\draw[->] (p1) to (e1);
\draw[->] (e1) to (p2);
\draw[->] (p2) to (e3);
\draw[->] (e3) to (p4);
\draw[->] (p1) to (e2);
\draw[->] (e2) to (p3);
\draw[->] (p3) to (e4);
\draw[->] (e4) to (p5);
\draw[->] (p4) to (e5);
\draw[->] (e5) to (p6);

\node[] (name) at (0,-6) {$\onl$};

\end{tikzpicture}}}
  \hspace{.25cm}
  \subfigure{\scalebox{.65}{\begin{tikzpicture}

\node[tplace, label=left:${[p_1]_{\sim_1}}$] (p1) at (0,-.75) {};
\node[transition] (e1) at (-.75,-1.5) {$a$};
\node[transition] (e2) at (.75,-1.5) {$b$};
\node[tplace, label=left:${[p_2]_{\sim_1}}$] (p2) at (-.75,-2.25) {};
\node[tplace, label=right:${[p3]_{\sim_1}}$] (p3) at (.75,-2.25) {};
\node[transition] (e4) at (.75,-3) {$d$};
\node[tplace, label=left:${[p_4]_{\sim_1}}$] (p4) at (-.75,-3) {};
\node[tplace, label=right:${[p_5]_{\sim_1}}$] (p5) at (.75,-3.75) {};
\node[transition] (e5) at (-.75,-3.75) {$c$};
\node[tplace, label=left:${[p_6]_{\sim_1}}$] (p6) at (-.75,-4.5) {};

\draw[->] (p1) to (e1);
\draw[->] (e1) to (p2);
\draw[->] (p2) to (e2);
\draw[->] (e2) to (p4);
\draw[->] (p1) to (e2);
\draw[->] (e2) to (p3);
\draw[->] (p3) to (e4);
\draw[->] (e4) to (p5);
\draw[->] (p4) to (e5);
\draw[->] (e5) to (p6);

\node[] (name) at (0,-5.25) {$\onl^{\sim_1}$};

\end{tikzpicture}}}
  \hspace{.25cm}
  \subfigure{\scalebox{.65}{\begin{tikzpicture}

\node[tplace, label=left:${p_2 \in [p_1]_{\sim_2}}$] (p1) at (0,-.75) {};
\node[transition] (e1) at (-.75,-1.5) {$a$};
\node[transition] (e2) at (.75,-1.5) {$b$};
\node[tplace, label=right:${[p3]_{\sim_2}}$] (p3) at (.75,-2.25) {};
\node[transition] (e4) at (.75,-3) {$d$};
\node[tplace, label=left:${[p_4]_{\sim_2}}$] (p4) at (-.75,-2.25) {};
\node[tplace, label=right:${[p_5]_{\sim_2}}$] (p5) at (.75,-3.75) {};
\node[transition] (e5) at (-.75,-3) {$c$};
\node[tplace, label=left:${[p_6]_{\sim_2}}$] (p6) at (-.75,-3.75) {};

\draw[->, bend right=15] (p1) to (e1);
\draw[->, bend right=15] (e1) to (p1);
\draw[->] (p1) to (e2);
\draw[->] (e2) to (p4);
\draw[->] (e2) to (p3);
\draw[->] (p3) to (e4);
\draw[->] (e4) to (p5);
\draw[->] (p4) to (e5);
\draw[->] (e5) to (p6);

\node[] (name) at (0,-4.5) {$\onl^{\sim_2}$};

\end{tikzpicture}}}
  \hspace{.25cm}
  \subfigure{\scalebox{.65}{\begin{tikzpicture}

\node[tplace, label=left:${p_2 \in [p_1]_{\sim_3}}$] (p1) at (0,-.75) {};
\node[transition] (e1) at (-.75,-1.5) {$a$};
\node[transition] (e2) at (.75,-1.5) {$b$};
\node[tplace, label=right:$p_4 \in {[p3]_{\sim_3}}$] (p3) at (.75,-2.25) {};
\node[transition] (e4) at (.75,-3) {$d$};
\node[tplace, label=right:${[p_5]_{\sim_3}}$] (p5) at (.75,-3.75) {};
\node[transition] (e5) at (-.75,-3) {$c$};
\node[tplace, label=left:${[p_6]_{\sim_3}}$] (p6) at (-.75,-3.75) {};

\draw[->, bend right=15] (p1) to (e1);
\draw[->, bend right=15] (e1) to (p1);
\draw[->] (p1) to (e2);
\draw[->] (e2) to (p3);
\draw[->] (p3) to (e4);
\draw[->] (e4) to (p5);
\draw[->] (p3) to (e5);
\draw[->] (e5) to (p6);

\node[] (name) at (0,-4.5) {$\onl^{\sim_3}$};

\end{tikzpicture}}}
  \caption{Folding equivalences and folded nets.}
\label{fig:fold}
\end{figure}

\begin{example}
\label{ex:fold1}
  Consider the log $\logs = \{ abc, bd \}$ and the independence relation $\ind = \emptyset$. \autoref{fig:fold} shows the obtained unfolding $\onl$ (left) and three of its folded nets. The equivalence relation $\sim_1$ merges events labeled by $b$, but it does not merge their presets, i.e. is not a SP folding equivalence. It can be observed that $bd$ is not fireable in $\onl^{\sim_1}$. Whenever two  events are merged, their preconditions need to be merged to preserved sequential executions. The equivalence relation $\sim_2$ does not only merge events labeled by $b$, but it also sets $p_1 \sim_2 p_2$ and is a SP folding equivalence. The folded net $\onl^{\sim_2}$ can replay every trace in the log $\logs$, but it also adds new traces of the form $a^*, a^*b, a^*bc, a^*bd, a^*bcd$ and $a^*bdc$.
\end{example}

Given an unfolding, every SP folding equivalence generates a net that preserves its sequential executions.

\begin{restatable}{theorem}{thetwo}
\label{the:fire_seq}
  Let $\on$ be an occurrence net and $\sim$ a SP folding equivalence, then every fireable sequence $M_0 \arrow{e_1} \dots \arrow{e_n} M_n$ from $\on$ generates a fireable sequence $[M_0]_\sim \arrow{[e_1]_\sim} \dots \arrow{[e_n]_\sim} [M_n]_\sim$ from $\on^\sim$.
\end{restatable}

\iftoggle{long}{
\begin{proof}
We reason inductively on the length of the fireable sequence.
\begin{description}
  \item[Base case:] if $n=0$, the results holds since an empty sequence of events from $\on$ generates an empty sequence of transitions that is trivially a fireable sequence from $\on^\sim$.
  \item[Inductive case:] we assume every fireable sequence $M_0 \arrow{e_1} \dots \arrow{e_n} M_n$ from $\on$ generates a fireable sequence $[M_0]_\sim \arrow{[e_1]_\sim} \dots \arrow{[e_n]_\sim} [M_n]_\sim$ from $\on^\sim$; we need to prove that the sequence $M_0 \arrow{e_1} \dots \arrow{e_{n+1}} M_{n+1}$ generates a fireable sequence $[M_0]_\sim \arrow{[e_1]_\sim} \dots \arrow{[e_{n+1}]_\sim} [M_{n+1}]_\sim$. Consider a fireable sequence $e_1 \dots e_{n+1}$ from $\on$, then by the inductive hypothesis, we know that the first $n$ events generate a firing sequence $[e_1]_\sim \dots  [e_n]_\sim$ from $\on^\sim$ leading to the marking $[M_n]_\sim$; we need to prove that $[M_n]_\sim \arrow{[e_{n+1}]}$. Suppose this is not true, then there exists $[b]_\sim \in \preset{[e_{n+1}]_\sim}$ with $[b]_\sim \not \in [M_n]_\sim$; the latter implies $b \not \sim b_n$ for every $b_n \in M_n$. \autoref{def:foldednet} does not add a flow arrow from a place $[b]_\sim$ to a transition $[e]_\sim$ unless a flow arrow exists 
between a condition $b' \in [b]_\sim$ and an event $e' \in [e]_\sim$ in the occurrence net. Therefore $[b]_\sim \in \preset{[e_{n+1}]_\sim}$ implies there exists $b_1 \sim b$ and $e'_{n+1} \sim e_{n+1}$ such that $b_1 \in \preset{e'_{n+1}}$ (there exists a flow arrow between $b_1$ and $e'_{n+1}$ in $\on$), and from the transitivity of $\sim$
follows that
  \begin{equation} \label{proof} \tag{*}
  b_1 \not \sim b_n \text{ for every } b_n \in M_n.
    \end{equation}
Since $\on$ allows a firing sequence of length $n+1$, we know $M_n \arrow{e_{n+1}}$ and then $\forall b_2 \in \preset{e_{n+1}}: b_2 \in M_n$. As every $b_2$ in $\preset{e_{n+1}}$ is also in $M_n$, by (\ref{proof}) we have $b_1 \not \sim b_2$ for all $b_2 \in \preset{e_{n+1}}$. From $b_1 \in \preset{e'_{n+1}}, e_{n+1} \sim e'_{n+1}$ and the fact that $\sim$ is a SP folding equivalence (\autoref{def:fold1}) follows that there exists $b_2 \in \preset{e_{n+1}}$ such that $b_1 \sim b_2$, but we showed this is not possible; therefore our assumption was false and for all $[b]_\sim \in \preset{[e_{n+1}]_\sim}$ we have $[b]_\sim \in [M_n]_\sim$. Finally $[M_n]_\sim \arrow{[e_{n+1}]_\sim}$ and $[M_0]_\sim \arrow{[e_1]_\sim} \dots \arrow{[e_{n+1}]_\sim} [M_{n+1}]_\sim$ is a fireable sequence from $\on^\sim$.
\qed
\end{description}
\end{proof}
}

As a corollary of the result above and \autoref{prop:fit}, the folded net obtained from $\onl$ with a SP folding
equivalence is fitting.

\begin{corollary}
\label{cor:fitting}
  Let $\logs$ be a log, $\ind$ an independence relation and $\sim$ a SP folding equivalence, then for every $\sigma \in \logs$ we have $\sigma \in \obs \netl$.
\end{corollary}

\iftoggle{long}{
\begin{proof}
  Since by \autoref{lemma:fit} every $\sigma \in \logs$ corresponds to a fireable sequence in $\onl$, the results follows immediately from \autoref{the:fire_seq}.
\end{proof}
}

\begin{example}
  We saw in \autoref{ex:fold1} that every trace from $\logs$ can be replayed in $\onl^{\sim_2}$, but (as expected) the net accepts more traces. However this net also adds some independence between actions of the system: after firing $b$ the net puts tokens at $[p_3]_{\sim_2}$ and $[p_4]_{\sim_2}$ and the reached marking enables concurrently actions $c$ and $d$ which contradicts $c \dep d$ (the independence relation $\ind = \emptyset$ implies $c \dep d$).
  In order to avoid this extra independence, we now consider the following class
  of equivalences.
\end{example}

\begin{definition}[Independence-preserving folding equivalence]
\label{def:fold2}
  Let $\on$ be an occurrence net and $\ind$ an independence relation; an equivalence relation $\sim$ is called an independence preserving (IP) folding equivalence iff
  \begin{enumerate}
    \item $\sim$ is a SP folding equivalence,
    \item $\lambda(e_1) \ind \lambda(e_2) \Leftrightarrow [\preset e_1]_\sim \cap [\preset{e_2}]_\sim = \emptyset \land [\preset e_1]_\sim \cap [\postset{e_2}]_\sim = \emptyset \land [\postset e_1]_\sim \cap [\preset{e_2}]_\sim = \emptyset$ for all events $e_1,e_2 \in E$.
    \item $b_1 \co b_2$ implies $b_1 \not \sim b_2$ for all conditions $b_1,b_2 \in B$.
  \end{enumerate}
\end{definition}

IP folding equivalences not only preserve the sequential behavior of $\on$, but also
ensure that $\on^\sim$ and $\on$ exhibit the same natural independence relation.

The definition above differs from the folding equivalence definition given in \cite{FahlandA13}; they
consider occurrence nets coming from an unfolding procedure which takes as an input a net. This procedure generates a
mapping between conditions and events of the generated occurrence net and places and transitions in the original net.
Such mapping is necessary to define their folding equivalence. In our setting, the occurrence net does not come from a
given net and therefore the mapping is not available. 

\begin{example}
  The equivalence $\sim_2$ from \autoref{fig:fold} is not an IP folding equivalence since the intersection of the equivalent classes of the preset of $c$ and $d$ is empty ($[\preset c]_{\sim_2} = \{[p_4]_{\sim_2}\}, [\preset d]_{\sim_2} = \{[p_3]_{\sim_2}\}$ and $\{[p_4]_{\sim_2}\} \cap \{[p_3]_{\sim_2}\} = \emptyset$), but $c$ and $d$ are not independent. Consider the equivalence relation $\sim_3$ which merges events labeled by $b$ and it sets $p_1 \sim_3 p_2$ and $p_3
\sim_3 p_4$; this relation is an IP folding equivalence. It can be observed in the net $\onl^{\sim_3}$ of
\autoref{fig:fold} that all the traces from the log can be replayed, but new independence relations are not introduced.
\end{example}

The occurrence net $\onl$ is clearly safe.
We show that $\netl$ is also safe when~$\sim$ is an IP folding equivalence.
In this work, we constraint IP equivalences to generate safe nets because their
natural independence relation is well understood (\autoref{def:indu}),
thus allowing us to assign a solid meaning to the class IP.
It is unclear what is the natural unconditional independence of an unsafe net,
and extending our definitions to such nets is subject of future work.

\begin{proposition}
\label{prop:safe}
  Let $\onl$ be the unfolding obtained from the log $\logs$ with $\ind$ as the independence relation and $\sim$ an IP folding equivalence. Then $\netl$ is safe.
\end{proposition}

\begin{proof}
  The unfolding $\onl$ is trivially safe since its initial marking puts one token in its minimal conditions and each condition contains only one event in its preset and that event cannot put more than one token in the condition. Suppose $\netl$ is not safe, by the above this is possible iff there exists $C \in \reach \onl$ and $b_1,b_2 \in C$ such that $b_1 \sim b_2$. If $b_1$ and $b_2$ belong to a reachable marking, then they must be concurrent and since $\sim$ is an IP folding equivalence they cannot be merged, which leads to a contradiction. Finally $\netl$ must be safe.
\qed
\end{proof}

\autoref{prop:ind_on} shows that the structural relation between events of the unfolding and the relation generated by the independence given by the expert coincide (when we restrict to co-enabled events); the result also holds for the folded net when an IP folding equivalence is used.

\begin{restatable}{theorem}{thethree}
\label{the:ip}
  Let $\onl$ be the unfolding obtained from the log $\logs$ with $\ind$ as the independence relation and $\sim$ an IP folding equivalence, then ${\indt \netl} = {\indu \netl}$.
\end{restatable}

\iftoggle{long}{
\begin{proof}
Let $(t,t') \in {\indt \netl}$, from \autoref{def:indt} this is true iff $\lambda(t) \ind \lambda(t')$ which is true iff for all $e \in t, e' \in t'$ we have $\lambda(e) \ind \lambda(e')$ (since the folding equivalence preserves labeling). As $\sim$ is a IP folding equivalence, independence between labeles holds iff for all $e\in t, e'\in t'$ we have $[\preset e]_\sim \cap [\preset{e'}]_\sim = \emptyset$ (see \autoref{def:fold2}.2). Using \autoref{def:foldednet}, the presets of $t$ and $t'$ are generated by some of the conditions in the preset of each $e$ and $e'$ respectively (the folding procedure does not introduces flow arrows) and we showed above that those conditions generate places that do not intersect those places generated by conditions in the preset of every $e'$; thus $[\preset e]_\sim \cap [\preset{e'}]_\sim = \emptyset$ iff $\preset{[e]_\sim} \cap \preset{[e']_\sim} = \emptyset$ iff $\preset t \cap \preset{t'} = \emptyset$ from $t = [e]_\sim$ and $t' = [e']_\sim$. Using the same 
reasoning it can be shown that independence between labels holds iff $\preset t \cap \postset{t'} = \emptyset$ and $\postset t \cap \preset{t'} = \emptyset$. Finally from \autoref{def:indu} we get $\preset t \cap \preset{t'} = \emptyset \land \preset t \cap \postset{t'} = \emptyset \land \postset t \cap \preset{t'} = \emptyset$ iff $(t,t') \in {\indu \netl}$.
\qed
\end{proof}
}

\subsection{Controlling Generalization via Negative Information}

We have shown that IP folding equivalences preserve independence.
However, they could still introduce new unintended behaviour not present
in~$\sys$.
In this section we limit this phenomena
by considering \emph{negative information},
denoted by traces that should not be allowed by the model.
Concretely, we consider negative information which is also given in the form of sequences $\sigma \in \logs^-
\subseteq A^*$. Negative information is often provided by an expert, but it can also be obtained
automatically by recent methods~\cite{BrouckeWVB14}. Very few techniques in the literature use negative
information in process discovery~\cite{Goedertier2009}. In this work, we assume a minimality criterion on the negative traces
used:


\begin{assumption}
  Let $\logs \eqdef \logs^+ \uplus \logs^-$ be a pair of positive and negative logs and $\ind$ the independence relation 
  given by the expert. Any negative trace $\sigma \in \logs^-$ corresponds to
the local configuration of some event $\esig$ in $\onl$.
\end{assumption}

This assumption implies that each negative trace is of the form $\sigma' a$ where $\sigma'$ only contains the actions
that are necessarily to fire $a$.  If $a$ can happen without them, they should not be consider part of $\sigma$. By
removing all events $\esig$ from $\onl$ (one for each
negative trace $\sigma \in \logs^-$), we obtain a new occurrence net denoted by $\onls$. The goal of this section is to fold this occurrence net without re-introducing the negative traces in the generalization step.
If the expert is unable to provide negative traces satisfying this assumption,
the discovery tool can always let him/her choose $\esig$ from a visual
representation of the unfolding.

\begin{definition}[Removal-aware folding equivalence]
\label{def:fold3}
  Let $\ondef$ be an occurrence net and $\logs^-$ a  negative log; an equivalence relation $\sim$ is called
removal
aware (RA) folding equivalence iff 
  \begin{enumerate}
    \item $\sim$ is a SP folding equivalence, and
    \item for every $\sigma \in \logs^-$ and $e' \in E$ we have $\lambda(e') = \lambda(\esig)$ implies $[\preset{e'}]_\sim \not \subseteq [\preset{\esig}]$.
  \end{enumerate}
\end{definition}

The folded net obtained from $\onls$ with a RA folding equivalence does not contain any of the negative traces.

\begin{restatable}{theorem}{thefour}
\label{the:ra}
  Let $\onls$ be the unfolding obtained from the log $\logs \eqdef \logs^+ \uplus \logs^-$ with $\ind$ as the independence
relation after removing the corresponding event of each negative trace and
$\sim$ a RA folding equivalence,\footnote{Since \autoref{def:fold3} refers to the events that generates the
local configurations of the negative traces, the folding equivalence must be defined over the nodes of $\onl$ and not
those of $\onls$.} then $$\obs \netls \cap \logs^- = \emptyset$$
\end{restatable}

\iftoggle{long}{
\begin{proof}
  Let $\sigma \eqdef \sigma' a$ and suppose $\sigma \in \obs \netls \cap \logs^-$. Since $\sigma \in \logs^-$, by \autoref{lemma:fit} $\sigma \in \obs \onl$, it follows by construction (see \autoref{def:lpo2es}) that $\sigma$ generates a unique local configuration which is removed in $\onls$ (by removing $\esig$). Thus $\sigma \not \in \obs \onls$, but $\sigma' \in \obs \onls$ since only the maximal event $\esig$ of the local configuration is removed. Let $M$ be the marking reached in $\onls$ after $\sigma'$, we know (using \autoref{the:fire_seq}) that $\sigma'$ generates a firing sequence in $\netls$ which leads to the reachable marking $[M]_\sim$. Since we assumed $\sigma \in \obs \netls$, there exists a transition $[e_a]_\sim$ such that $[M]_\sim \arrow{[e_a]_\sim}$ with $\lambda([e_a]_\sim) = a$, but this implies (from \autoref{def:foldednet}) that the preset of $\esig$ was merged with the preset of $e_a$ which contradicts the assumption that $\sim$ is a RA folding equivalence. Finally the assumption was 
false and $\obs \netls \cap \logs^- = \emptyset$. \qed
\end{proof}
}

\section{Computing Folding Equivalences}
\label{sec:computing}

\autoref{sec:discovery} presents a discovery algorithm that generates fitting
occurrence nets and \autoref{sec:generalization} defines three classes of
folding criteria, SP, IP, and RA, that ensure various properties.
This section proposes an approach to synthesize~SP,~IP and~RA folding equivalences
using~SMT.

\subsection{SMT Encoding}
\label{sec:sat}

We use an SMT encoding to find folding equivalences generating a net $\on^\sim$ satisfying specific metric properties.
Specifically, given a measure $\hat c$ (cf., \autoref{sec:prelim}),
decidable in polynomial time, and a number~$k \in \nat$, we generate an SMT formula
which is satisfiable iff there exists a folding equivalence $\sim$ such that $\hat c(\on^\sim) = k$. 
We consider the number of transitions in the folded net as the measure $\hat c$,
however, theoretically, any other measure that can be computed in polynomial time
could be used.
As explained in \autoref{sec:prelim}
simple functions like counting the number of nodes/arcs provide in
practice reasonable results.

Given an occurrence net~$\ondef$,
for every event $e \in E$ and condition $b \in B$
we have integer variables $v_e$ and $v_b$.
The key intuition is that two events (conditions) whose variables have equal
number are equivalent and will be merged into the same transition (place).
The following formulas state, respectively, that every element of a set $X$ is
related with at least one element of a set $Y$, and that every element of $X$ is
not related with any element of $Y$: 
$$\ssc X Y \eqdef \bigwedge\limits_{x \in X}\bigvee\limits_{y \in Y} (v_x = v_y)
\hspace{10mm} \disj X Y \eqdef \hspace{-2mm}\bigwedge\limits_{x \in X, y \in Y} \hspace{-2mm} (v_x \not = v_y)$$

We force any satisfying assignment to represent an~SP folding equivalence
(\autoref{def:fold1}) with the following two constraints:
$$
\phi_\on^{SP} \eqdef \phi_\on^{lab} \land \phi_\on^{pre}.
$$
Formulas $\phi_\on^{lab}$ and $\phi_\on^{pre}$ impose that only equally
labeled events should be equivalent and that if two events are equivalent, then their presets
should generate the same equivalence class:
$$\phi_\on^{lab} \eqdef \bigwedge\limits_{\substack{e,e' \in E \\ \lambda(e)
\not = \lambda(e')}} \hspace{-2mm} (v_e \not = v_{e'}) \hspace{10mm}
\phi_\on^{pre} \eqdef \bigwedge\limits_{e,e' \in E} (v_e = v_{e'} \Rightarrow (\ssc {\preset e} {\preset{e'}} \land \ssc {\preset{e'}} {\preset e}))$$

In addition to the properties encoded above, an IP folding equivalence
(\autoref{def:fold2}) should satisfy some other restrictions: $$\phi_\on^{IP}
\eqdef \phi_\on^{SP} \land \phi_\on^{ind} \land \phi_\on^{co}$$
where $\phi_\on^{ind}$ imposes that the presets and postsets of events with
independent labels should generate equivalence classes that do not intersect and
$\phi_\on^{co}$ forbids concurrent conditions to be merged: $$\phi_\on^{ind} \eqdef
\bigwedge\limits_{e,e' \in E} \hspace{-2mm} (\lambda(e) \ind \lambda(e')
\Leftrightarrow (\disj{\preset e}{\preset{e'}} \land \disj{\preset
e}{\postset{e'}} \land \disj{\preset e}{\postset{e'}}))
\hspace{8mm}\phi_\on^{co} \eqdef \bigwedge\limits_{\substack{b,b' \in B \\ b \co b'}} \hspace{-2mm} (v_b \not = v_{b'})$$


Given a negative log $\logs^-$, to encode a RA folding equivalence
(\autoref{def:fold3}) we define: $$\phi_{\on, \logs^-}^{RA} \eqdef \phi_\on^{SP} \land  
(\hspace{-2mm}\bigwedge\limits_{\substack{\sigma \in \logs^-, e' \in E\\ \lambda(e') = \lambda(\esig)}} \hspace{-3mm}
\neg \ssc {\preset{e'}} {\preset \esig})$$ where the right part of the conjunction imposes that for every $\esig$
generated by a negative trace and any other event with the same label, their presets cannot generate the same equivalence
class.

We now encode the optimality (\wrt the number of transitions) of the mined net.
Given an occurrence net~$\ondef$, each event $e \in E$ generates a
transition $v_e$ in the folded net $\on^\sim$. To impose that the number of
transitions in $\on^\sim$ should be at most $k \in \nat$, we define:
$$\phi_{\on,k}^{MET} \eqdef \bigwedge\limits_{e \in E} (1 \le v_e \leq k)$$

To find an IP and RA folding equivalence that generates a net with at most $k$
transitions we propose the following encoding: $$\phi_{\on,\logs^-,k}^{OPT}
\eqdef \phi_{\on}^{IP} \land \phi_{\on,\logs^-}^{RA} \land \phi_{\on,k}^{MET}$$

\begin{restatable}{theorem}{thefive}
  Let $\logs \eqdef \logs^+ \uplus \logs^-$ be a set of positive and negative logs, $\ind \subseteq A \times A$ and independence relation and $k \in \nat$. The formula $\phi_{\on,\logs^-,k}^{OPT}$ is satisfiable iff there exists an IP and RA folding equivalence $\sim$ such that $\on_{\logs, \ind, *}^\sim$ contains at most $k$ transitions.
\end{restatable}

\iftoggle{long}{
\begin{proof}
  Let $\psi$ be a solution of $\phi_{\on,\logs^-,k}^{OPT}$ and let $\sim_\psi$ be the relation such that $x \sim_\psi x'$ iff $\psi \models (v_x = v_{x'})$, i.e. $\psi$ assigns the same value to $v_x$ and $v_{x'}$. By the reflexivity, symmetry and transitivity of integer numbers follows that $\sim_\psi$ is an equivalence relation. The assignment $\psi$ is a solution of the formula iff all of the following are true:
  \begin{enumerate}
    \item $\phi_{\onl}^{IP}$ holds; this is true iff \emph{(i)} for every two events $e,e'$ with different labels $v_e \not = v_{e'}$, \emph{(ii)} if $v_e = v_{e'}$ then for all $b \in \preset e$ there exists $b'
\in \preset{e'}$ such that $v_b = v_{b'}$ and viceversa, \emph{(iii)} for every pair
$e,e'$ of events with independent labels \emph{(iii.a)} for all conditions $b \in \preset e,b' \in \preset{e'}$ we have $v_b \not = v_{b'}$, \emph{(iii.b)} for all conditions $b \in \preset e,b' \in \postset{e'}$ we have $v_b \not = v_{b'}$, \emph{(iii.c)} for all conditions $b \in \postset e,b' \in \preset{e'}$ we have $v_b \not = v_{b'}$, \emph{(iv)} for every pair $b,b'$ of concurrent conditions we have $v_b \not \sim v_{b'}$; by the definition of $\sim_\psi$ we have \emph{(i)} for every two events
$e,e'$ with different labels $e \not \sim_\psi e'$, \emph{(ii)} if $e \sim_\psi e'$ then $[\preset e]_{\sim_\psi}
= [\preset{e'}]_{\sim_\psi}$, \emph{(iii)} for every pair
$e,e'$ of events with independent labels $[\preset e]_{\sim_\psi} \cap [\preset{e'}]_{\sim_\psi} = \emptyset, [\preset e]_{\sim_\psi} \cap [\postset{e'}]_{\sim_\psi} = \emptyset$ and $[\postset e]_{\sim_\psi} \cap [\preset{e'}]_{\sim_\psi} = \emptyset$, \emph{(iv)} $b \co b'$ implies $b \not \sim b'$; by \autoref{def:fold2} this is true iff the relation $\sim_\psi$
is an IR folding equivalence. 
    \item $\phi_{\onl, \logs^-}^{RA}$ holds; this is true iff \emph{(i)} for every two events events $e,e'$ with different labels $v_e \not = v_{e'}$, \emph{(ii)} if $v_e = v_{e'}$ then for all $b \in \preset e$ there exists $b' \in \preset{e'}$  such that $v_b = v_{b'}$ and viceversa, \emph{(iii)} for any trace $\sigma \in \logs^-$ and any event $e' \in E$ with the same label as $\esig$ there exists a conditions $b \in \preset{e'}$ such that for any condition $b' \in \preset \esig$ we have $v_b \not = v_{b'}$; by the definition of $\sim_\psi$ we have \emph{(i)} for every two events $e,e'$ with different labels $e \not \sim_\psi e'$, \emph{(ii)} if $e \sim_\psi e'$ then $[\preset e]_{\sim_\psi} = [\preset{e'}]_{\sim_\psi}$ and \emph{(iii)} $[\preset{e'}]_{\sim_\psi} \not \subseteq [\preset \esig]_{\sim_\psi}$ for any negative trace $\sigma$ and event $e'$ with the same label as $\esig$; by \autoref{def:fold3} this is true iff $\sim_\psi$ is a RA folding equivalence. 
    \item $\phi_{\onl,k}^{MET}$ holds, this is true iff $v_e \leq k$ for every event $e \in E$; the encoding associates a number to each equivalence class (according to $\sim_\psi$) of events and bounds the number of equivalence classes by $k$, since the number of transitions in $\onls^{\sim_\psi}$ corresponds to the number of equivalent classes of events (see \autoref{def:foldednet}), this is true iff the number of transitions of $\onls^{\sim_\psi}$ is bounded by $k$.
\qed
  \end{enumerate}
\end{proof}
}

\subsection{Finding an Optimal Folding Equivalence}
\label{sec:opt_fold}

\autoref{sec:sat} explains how to compute a folding equivalence that generates a folded net with a bounded number of transitions; this section explain how to obtain the optimal folded net, i.e the one with minimal number of transitions satisfying the properties of \autoref{the:ip} and \autoref{the:ra}.

Iterative calls to the SMT solver can be done for a binary search with $k$ between $min_k$ and $max_k$; since only equally labeled events can be merged by the folding equivalence, the minimal number of transitions in the folded net is $min_k \eqdef \lvert A \rvert$; in the worst case, when events cannot be merged, $max_k \eqdef \lvert E \rvert$.

As a side remark, we have noted that the optimal folding equivalence can be encoded as a MaxSMT
problem~\cite{NieuwenhuisO06} where some clauses which are called hard must be true in a solution (in our case
$\phi_{\on}^{IP}$ and $\phi_{\on,\logs^-}^{RA}$) and some soft clauses may not ($\phi_{\on,k}^{MET}$ for $\lvert A
\rvert \leq k \leq \lvert E \rvert$); a MaxSMT solver maximizes the number of soft clauses that are satisfiable and thus
it obtains the minimal $k$ generating thus the optimal folded net.


\section{Experiments}
\label{sec:experiments}

\newcommand\ratiosm{r_{{\sys} \subseteq {\model}}}
\newcommand\ratioms{r_{{\model} \subseteq {\sys}}}

As a proof of concept, we implemented our approach into
a new tool called \podtool (Partial Order
Discovery).\footnote{Tool and benchmarks: 
\url{http://lipn.univ-paris13.fr/~rodriguez/exp/atva15/}.}
It supports synthesis of SP and IP folding equivalences using a restricted form
of our SMT encoding.
In particular \podtool merges all events with equal label, in contrast to the
encoding in \autoref{sec:computing} which may in general yield more than
one transition per log action.
While this ensures a minimum (optimal as per \autoref{sec:opt_fold})
number of folded transitions, the tool could
sometimes not find a suitable equivalence (unsatisfiable SMT encoding).
Since the number of transitions in the folded net is fixed,
it turns out that the quality of the mined model increases as we increase the
number of folded places, as we show below.
Using \podtool
we evaluate the ability of our approach to rediscover the
original process model, given its independence relation and a set of logs.
For this we have used standard benchmarks from the verification and process mining literature~\cite{MCC,WDHS08}.

\newcommand\newrow{\\}

\begin{table}[t]

\setlength\tabcolsep{3.7pt}
\def\sep{\hspace{22pt}}
\def\tinysep{\hspace{4pt}}
\def\negsep{\hspace{2.5pt}}

\centering
\footnotesize
\tt
\begin{tabular}{lrr@{\sep}rr@{\negsep}rr@{\sep}rr@{\negsep}rr}
\toprule
  \multicolumn{3}{l}{\rm\small Original}
& \multicolumn{4}{l}{\rm\small \podtool (max.\ places)}
& \multicolumn{4}{l}{\rm\small \podtool (60\% places)}
\\
  \cmidrule(r){1-3}
  \cmidrule(r){4-7}
  \cmidrule(r){8-11}
  \rm\small Benchmark
& $\lvert T \rvert$
& $\lvert P \rvert$
& $\ratiosm$
& $\ratioms$
& \rm\small \%Prec.
& $\lvert P \rvert$
& $\ratiosm$
& $\ratioms$
& \rm\small \%Prec.
& $\lvert P \rvert$
\\
\midrule

\rm\bench[22]{A}     &  22 &  20 & 0.99 & 1.00 & 0.77 & 19 & 0.57 & 1.00 & 0.22 & 11 \newrow
\rm\bench[32]{A}     &  32 &  32 & 1.00 & 1.00 & 0.80 & 32 & 0.46 & 1.00 & 0.19 & 19 \newrow
\rm\bench[42]{A}     &  42 &  47 & 0.98 & 1.00 & 0.54 & 40 & 0.79 & 1.00 & 0.21 & 28 \newrow
\rm\bench[32]{T}     &  33 &  31 & 1.00 & 1.00 & 0.88 & 31 & 0.54 & 1.00 & 0.19 & 18 \newrow
\rm\bench[1]{Angio}  &  64 &  39 & 0.39 & 0.94 & 0.18 & 21 & 0.10 & 0.92 & 0.06 & 13 \newrow
\rm\bench{Complex}   &  19 &  13 & 0.98 & 1.00 & 0.62 & 12 & 0.62 & 1.00 & 0.39 & 7  \newrow
\rm\bench{ConfDimB}  &  11 &  10 & 1.00 & 1.00 & 1.00 & 10 & 0.62 & 1.00 & 0.39 & 6  \newrow
\rm\bench[5]{Cycles} &  20 &  16 & 1.00 & 1.00 & 1.00 & 16 & 0.60 & 1.00 & 0.40 & 6  \newrow
\rm\bench[2]{DbMut}  &  32 &  38 & 0.98 & 0.98 & 0.94 & 32 & 0.76 & 0.98 & 0.21 & 19 \newrow
\rm\bench{Dc}        &  32 &  35 & 0.99 & 0.99 & 0.77 & 27 & 0.84 & 0.99 & 0.38 & 21 \newrow
\rm\bench[2]{Peters}  & 126 & 102 & 0.45 & 1.00 & 0.07 & 51 & 0.30 & 1.00 & 0.05 & 30
\\
\bottomrule
\end{tabular}
\vspace{0pt}
\rm
\caption{Experimental results.}
\label{tab:exp}
\end{table}

%
%
%

In our experiments, \autoref{tab:exp},
we consider a set of original processes faithfully modelled as safe Petri nets.
For every model $\sys$ we consider a log~$\logs$, i.e. a subset of its traces.
We extract from~$\sys$
the (best) independence relation~$\indu \sys$ that an expert could provide.
We then provide~$\logs$ and~$\indu \sys$ to \podtool and
find an SP folding equivalence with the largest number
of places (\cols ``max.\ places'') and
with~60\% of the places of~$\sys$ (last group of \cols),
giving rise to two different mined models.
All three models, original plus mined ones, have perfect fitness but varying
levels of precision, i.e. traces of the model not present in the log.
For the mined models, we report (\cols ``\%Prec.'')
on the ratio between their precision and the precision of the original model~$\sys$.
All precisions were estimated using the technique from~\cite{AMCDA15}.
All \podtool running times were below 10s.

Additionally, we measure how much independence of the original model is
preserved in the mined ones.
For that, we define the ratios~$\ratiosm \eqdef
\nicefrac{|{\indu \sys} \cap {\indu \model}|}{|\indu \sys|}$
and~$\ratioms \eqdef
\nicefrac{|{\indu \sys} \cap {\indu \model}|}{|\indu \model|}$.
The closer $\ratiosm$ is to~1, the larger is the number of
pairs in $\indu \sys$ also contained in $\indu \model$
(\ie, the more independence was preserved),
and conversely for $\ratioms$
(the less independence was ``\emph{invented}'').
Remark that
${\indu \sys} = {\indu \model}$ iff
$\ratiosm = \ratioms = 1$.


%
%
%

In 7 out of the 11 benchmarks in \autoref{tab:exp} our proof-of-concept tool
rediscovers the original model or finds one with only minor differences.
This is even more encouraging when
considering that we only asked \podtool to find SP equivalences
which, unlike IP, do not guarantee preservation of independence.
In 9 out of 11 cases both ratios $\ratiosm$ and $\ratioms$ are above 98\%,
witnessing that independence is almost entirely preserved.
Concerning the precision,
we observe that it is mostly preserved for these 9 models.
We observe a clear correlation between the number of discovered places and the
precision of the resulting model.
The running times of \podtool on all benchmarks in \autoref{tab:exp}
were under few seconds.

In \bench[2]{Peters} and \bench[1]{Angio} our tool could not increase the
number of places in the folded net, resulting in a significant loss of
independence and precision. We tracked the reason down to
(a) the additional restrictions on the SMT encoding imposed by our implementation and
(b) the algorithm for transforming event structures
into unfoldings (\ie, introducing conditions).
We plan to address this in future work.
This also prevented us from of employing IP equivalences instead of SP for
these experiments: \podtool could find IP equivalences for only~5 out of~11 cases.
Nonetheless, as we said before, in 9 out of 11 the found SP
equivalences preserved at least 98\% of the independence.

Finally,
we instructed \podtool to synthesize~SP equivalences folding into an arbitrarily
chosen low number of places (60\% of the original).
Here we observe a large reduction of precision and significant loss of
independence (surprisingly only~$\ratiosm$ drops, but not~$\ratioms$).
This witnesses a strong dependence between the number of discovered places
and the ability of our technique to preserve independence.

\section{Related Work}

To the best of our knowledge, there is no technique in the literature that solves the particular problem we are
considering in this paper: given a set of positive and negative traces and an independence relation on events, derive a
Petri
net that both preserves the independence relation and satisfies the quality dimensions enumerated in
\autoref{sec:prelim}.
However, there is
related work that intersects partially with the techniques of this paper.
We now report on it.

Perhaps the closest work is~\cite{FahlandA13}, where the simplification of an initial process model is done by first
unfolding the model (to derive an overfitting model) and then folding it back in a controlled manner, thus
generalizing part of the behavior. The
approach can only be applied for fitting models, which hampers its applicability unless alignment
techniques~\cite{AryaThesis} are used. 
The folding equivalences presented
in this paper do not consider a model and therefore are less restrictive than the ones presented in~\cite{FahlandA13}. 

{\em Synthesis} is a  problem different from discovery: in synthesis, the underlying system is given and therefore
one can assume $\sys = \logs$. Considering a synthesis scenario, Bergenthum {\em et al.} have investigated the
synthesis of a p/t net from partial orders~\cite{BergenthumDLM08}. The class of nets considered in this paper (safe
Petri nets) is less expressive than p/t nets, which in practice poses no problems in the context of business processes.
The algorithms in~\cite{BergenthumDLM08} are grounded in the {\em theory of regions} and split the problem into two
steps  \emph{(i)} the p/t net $\model$ is generated which, by construction, satisfies $\logs \subseteq \obs {\model}$,
and \emph{(ii)} it is checked whereas $\logs = \obs {\model}$. Actually, by avoiding \emph{(ii)}, a discovery scenario
is obtained where the generalization feature is not controlled, in contrast to the technique of this paper. With the
same goal but relying on ad-hoc operators tailored to compose lpos (choice, sequentialization, parallel compositions and 
repetition), a discovery technique is presented in~\cite{BergenthumDML09}. Since the operators may in practice introduce
wrong generalizations, a domain expert is consulted for the legality of every extra run.

\section{Conclusions}

A fresh look at process discovery is presented in this paper, which establishes
theoretical basis for coping with some of
the challenges in the field. By automating the folding of the unfolding that covers traces in the log but also
combinations thereof derived from the input independence relation, problems like log incompleteness and noise may be
alleviated. The approach has been implemented and the initial results show the potential of the technique in
rediscovering a model, even for the simplest of the folding equivalences described in this paper.

Next steps will focus on implementing the remaining folding equivalences,
and in general improving the SMT constraints for computing folding equivalences.
Also, incorporating the notion of trace frequency in the approach
will be considered, to guide the technique to focus on principal behavior. This will allow to also test the tool in
presence of incomplete or noisy logs.

\bibliographystyle{plain}

\bibliography{mybib}

\end{document}